\documentclass[pra,onecolumn,groupedaddress,superscriptaddress,amsmath,amssymb,longbibliography]{revtex4-2}
\usepackage{graphicx}
\usepackage{mathrsfs}
\usepackage{amsmath,amssymb,amsthm,bm}
\usepackage[T1]{fontenc}
\usepackage{xcolor}
\usepackage[normalem]{ulem}

\usepackage{hyperref}
\usepackage[nameinlink,capitalize]{cleveref}
\hypersetup{colorlinks=true,linkcolor=purple,citecolor=purple,urlcolor=blue}

\newtheorem{proposition}{Proposition}
\newtheorem{theorem}{Theorem}
\newtheorem{lemma}{Lemma}
\newtheorem{corollary}{Corollary}
\theoremstyle{definition}
\newtheorem{remark}{Remark}

\DeclareMathOperator{\Tr}{Tr}
\DeclareMathOperator{\spanop}{span}
\DeclareMathOperator{\Ran}{Ran}

\newcommand{\ket}[1]{|#1\rangle}
\newcommand{\bra}[1]{\langle #1|}

\usepackage{orcidlink}
\begin{document}

\title{Sparse anisotropic positive maps for qutrit entanglement: exact indecomposability and PPT geometry}

\author{Massimiliano F. Sacchi\,\orcidlink{0000-0002-8909-2196}}
\email[Massimiliano F. Sacchi: ]{massimiliano.sacchi@cnr.it}
	\affiliation{CNR-Istituto di Fotonica e Nanotecnologie, Piazza Leonardo da Vinci 32, I-20133, Milano, Italy}
	\affiliation{Dipartimento di Fisica, Università degli Studi di
          Pavia, Via Agostino Bassi 6, I-27100, Pavia, Italy}


\begin{abstract}
Positive but not completely positive maps provide one of the most direct
ways to detect entanglement beyond the positive-partial-transpose (PPT)
criterion.  We introduce and analyze an exactly solvable two-parameter
family of sparse bistochastic positive maps on qutrits, in which two
coherence channels are independently tuned by parameters $w$ and $z$.
The sparse structure makes the full phase diagram analytic: positivity
holds exactly on the square $0\le w,z\le2/3$, complete positivity on the
smaller square $0\le w,z\le1/3$, and decomposability is lost precisely
outside a quarter circle in the corner $w,z\ge1/3$.  The
indecomposable region is certified by explicit PPT entangled states
adapted to the same witness geometry.  At the endpoint
$W_*=W(2/3,2/3)$ we construct a four-parameter family of PPT edge states
of rank type $(5,5)$, derive their analytic detection region, and show
that the corresponding rays are exposed faces of the PPT cone.  Finally,
although $W_*$ is not optimal, we give an explicit optimal refinement
whose detection region on this family is strictly larger.  The result is
an analytically tractable qutrit setting in which positivity,
indecomposability, PPT entanglement, optimality, and exposed convex
geometry can be studied in a single framework.
\end{abstract}

\maketitle

\section{Introduction}
\label{sec:intro}
Entanglement is a defining nonclassical feature of quantum theory and a central
resource for quantum information processing.  Deciding whether a mixed
bipartite state is separable remains difficult, and positive maps provide one
of the most effective formulations of the problem: a state is separable if and
only if it remains positive under $\mathrm{id}\otimes\Phi$ for every positive
linear map $\Phi$ \cite{Horodecki1996,Horodecki2009}.  Through the
Choi--Jamio\l{}kowski correspondence, this criterion translates questions
about maps into spectral and block-positivity questions for bipartite
operators \cite{Jamiolkowski1972,Choi1975}.

Indecomposable positive maps are especially important because
decomposable maps cannot detect PPT entangled states, whereas
indecomposable maps can
\cite{Peres1996,Horodecki1996,Terhal2000,Lewenstein2000}.  In low
dimensions such maps are excluded by the results of St\o rmer and
Woronowicz \cite{Woronowicz1976,Stormer2013}, while in dimension
$3\otimes3$ they appear already in the Choi map and its
generalizations \cite{ChoKyeLee1992}.  The subsequent Ha--Kye analysis
of Choi-type witnesses, exposed maps, and PPT edge states shows that
positive maps give a concrete setting in which map geometry and state
geometry can be studied simultaneously
\cite{HaKye2004,HaKye2005,HaKyeExposed2011,HaKye2012,ChruscinskiSarbicki2014}.
Complementary developments relate positive maps and entanglement
witnesses to tensor-cone characterizations, spectral and geometric
constructions, positive biquadratic forms, elementary-operator
criteria, merging procedures, and tensor-power decomposability
\cite{MajewskiMarciniak2001,Kossakowski2003,ChruscinskiKossakowski2007,ChruscinskiKossakowski2009,SkowronekZyczkowski2009,Hou2010,MarciniakRutkowski2017,MullerHermes2018}. More
recent developments also use positive maps to build practical
separability criteria, entanglement witnesses, and
entanglement-estimation methods based on symmetric measurements,
moment criteria, realignment-like witness constructions, and symmetric
positive-map
constructions~\cite{LiYaoFeiFanMa2024,MallickMaityGangulyMajumdar2025,JannesaryKarimipourChruscinski2025,Rico2026}.
For quantum-information applications it is useful to have examples that are
not only indecomposable but also analytically tractable.  Such examples
serve as benchmarks for numerical searches for positive maps, for testing
entanglement criteria on PPT states, and for separating distinct notions such
as positivity, complete positivity, decomposability, optimality, and
exposedness.  The challenge is that these properties are usually governed by
different convex cones and are rarely accessible in closed form within one
and the same family.

The present paper provides such an exactly solvable construction.  We analyze a
family of sparse positive maps for qutrits introduced in
Ref.~\cite{CompanionSparseMaps}, focusing on the anisotropic regime in which two
surviving coherences have independent weights.  The Choi matrix has one
suppressed coherence sector and two active off-diagonal sectors.  This simple
pattern is the main physical and mathematical simplification: the map probes two
selected coherences of a bipartite qutrit system, while the missing sector makes
the positivity and PPT geometry sufficiently rigid to be solved exactly.

In a uniform bistochastic slice we obtain the following results.  First, the
positive maps occupy the square $[0,2/3]\times[0,2/3]$, whereas the completely
positive maps occupy the smaller square $[0,1/3]\times[0,1/3]$.  Second, the
transition from decomposable to indecomposable positive maps is exact: every
positive map with either $w\le1/3$ or $z\le1/3$ is decomposable, while in the
corner $w,z\ge1/3$ the boundary is the quarter circle
\begin{equation}
\left(w-\frac13\right)^2+\left(z-\frac13\right)^2=\frac19\;.\label{uno}
\end{equation}
Third, the indecomposable region is detected by analytic PPT entangled states
adapted to the same two coherence channels.  Fourth, at the endpoint witness
$W_*=W(2/3,2/3)$ we construct a four-parameter family of PPT edge states of rank
type $(5,5)$, determine their detection region, and prove that the corresponding
PPT-state rays are exposed.  Finally, we show that $W_*$ is nonoptimal and give
an explicit optimal refinement with a strictly larger analytic detection region
on the constructed states.

The construction therefore gives more than a phase diagram.  It gives a
controlled witness--state correspondence in which the same sparse qutrit
structure determines the positivity region, the decomposability threshold, the
PPT entangled states detected by the witness, and the relevant faces of the PPT
cone.

The paper is organized as follows.  In Sec.~\ref{sec:theory} we recall
the basic notions of positive maps, decomposability, and Choi
matrices.  In Sec.~\ref{sec:family} we introduce the qutrit family of
maps and specialize to the uniform bistochastic slice.
Section~\ref{sec:positivity} gives the positivity region. Section~\ref{sec:channels} identifies the physical
  quantum-channel sector and proves that its completely positive maps
  are all entanglement breaking. Sec.~\ref{sec:decomposability}
establishes the decomposability threshold using explicit PPT
witnesses.  In Sec.~\ref{sec:pptstates} we construct witness-adapted
PPT edge-state families and analyze their detection regions.  Finally,
Sec.~\ref{sec:wstar} collects the witness-side properties of the
endpoint map, including its dual face, nonoptimality, and an optimal
refinement with a larger detection region.
\section{Preliminaries and notation}
\label{sec:theory}

We recall the basic notions for Hermiticity-preserving linear maps between
matrix algebras and their relation to bipartite entanglement.

Let $M_d$ denote the algebra of complex $d\times d$ matrices, and let
$\Phi:M_d \to M_d$ be a linear map. We consider Hermiticity-preserving
linear maps, satisfying $\Phi[X^\dag ] =\Phi[X]^\dag $.  Among such
maps, the following classes play a central role:
\begin{enumerate}
\item $\Phi$ is \emph{positive} if $\Phi[X]\succeq0$ for every positive
semidefinite operator $X\in M_d$;
\item $\Phi$ is \emph{$k$-positive} if the ampliation $\mathrm{id}_k\otimes\Phi$ is positive on
$M_k \otimes M_d$; this hierarchy is also connected with Schmidt-number
witnesses and refined entanglement classification \cite{SanperaBrussLewenstein2001};
\item $\Phi$ is \emph{completely positive} if it is $k$-positive for
  every $k\ge1$;
\item $\Phi$ is \emph{completely co-positive} if $\Phi\circ {\cal T}$ is completely
positive, where $\mathcal T$ denotes transposition on $M_d$;
\item $\Phi$ is \emph{decomposable} if it can be written as the sum of a
completely positive map and a completely co-positive one.
\end{enumerate}

The main tool is the \emph{Choi matrix}
\begin{equation}
\label{eq:ChoiM}
C_\Phi=(\mathrm{id}_d\otimes\Phi)\ket{\psi}\bra{\psi}
=\sum_{n,m=1}^{d}\ket{n}\bra{m}\otimes\Phi[\ket{n}\bra{m}],
\end{equation}
where
\begin{equation}
\ket{\psi}=\sum_{n=1}^{d}\ket{n}\otimes\ket{n}
\in\mathbb{C}^d\otimes\mathbb{C}^d
\label{eq:sym-vec}
\end{equation}
denotes the unnormalized maximally entangled vector. Choi's theorem states that $\Phi$ is completely positive if
and only if $C_\Phi\succeq0$~\cite{Choi1975}. Equivalently, for maps
on $M_d$, $d$-positivity already implies complete
positivity.~\cite{Choi1975,Jamiolkowski1972}

Ordinary positivity is weaker than complete positivity.  A
Hermiticity-preserving map $\Phi$ is positive if and only if its Choi matrix is
\emph{block-positive}, namely if
\begin{equation}
\bra{\chi}\otimes\bra{\phi}\,C_\Phi\,\ket{\chi}\otimes\ket{\phi}\ge0
\qquad \forall\,\ket{\chi},\ket{\phi}\in\mathbb{C}^d.
\end{equation}
By convexity, this is equivalent to non-negativity on all separable
states.  Conversely, a bipartite state is entangled if and only if it
is detected by some positive but not completely positive
map~\cite{Horodecki1996}, meaning that
$(\mathrm{id}\otimes\Phi)[\rho]$ fails to be positive for at least one
such $\Phi$. In this sense, positive but not completely positive maps
act as entanglement witnesses.

We shall use the following PPT test for indecomposability (throughout
the paper, the partial transposition $\tau$ acts on the second tensor
factor unless explicitly stated otherwise).
\begin{theorem}
\label{thm:ppt-detects-indecomp}
Let $C_\Phi$ be the Choi matrix of a positive map $\Phi$. Suppose
there exists a bipartite state $\rho$ such that $\Tr[C_\Phi\rho]<0$.
Then $\rho$ is necessarily entangled. If, in addition, $\rho$ is PPT,
namely if $\rho^\tau\succeq0$, then $\Phi$ is indecomposable.
This standard criterion follows from the duality between decomposable witnesses and the PPT cone~\cite{Horodecki1996,Lewenstein2000}.
\end{theorem}
\begin{proof}
If $\Phi$ were decomposable, then its Choi matrix could be written as
\begin{equation}
C_\Phi=P+Q^\tau ,
\label{eq:decomp_choi}
\end{equation}
with $P,Q\succeq0$. For every PPT state $\rho$ one would then have
\begin{equation}
\Tr[C_\Phi\rho]=\Tr[P\rho]+\Tr[Q^\tau\rho]
=\Tr[P\rho]+\Tr[Q\rho^\tau]\ge0,
\label{eq:ppt_test_decomp}
\end{equation}
because both $\rho$ and $\rho^\tau$ are positive semidefinite. Hence the
existence of a PPT state with negative expectation value rules out every
decomposition of the form~\eqref{eq:decomp_choi}. The first assertion follows
from block-positivity of $C_\Phi$ on separable states.
\end{proof}

\section{Qutrit family and uniform bistochastic slice}
\label{sec:family}
We now specialize the sparse qutrit maps of Ref.~\cite{CompanionSparseMaps} to
the anisotropic regime $w\neq z$ and uniform bistochastic slice. Each map acts on $X\in M_3$ as
\begin{equation}
\Phi[X]=
\begin{pmatrix}
 aX_{11}+dX_{22}+bX_{33} & w\,X_{21} & z\,X_{13}\\
 w\,X_{12} & bX_{11}+aX_{22}+dX_{33} & 0\\
 z\,X_{31} & 0 & dX_{11}+bX_{22}+aX_{33}
\end{pmatrix},
\label{eq:map}
\end{equation}
with parameters $a,b,d\ge0$ and $w,z\ge0$. The associated Choi matrix has two
nontrivial $2\times2$ blocks and five one-dimensional sectors, so its spectrum
can be read off explicitly as~\cite{CompanionSparseMaps}
\begin{equation}
\operatorname{spec}(C_\Phi)=
\left\{
 a+z,\ a-z,\ a,\ b,\ b,\ d,\ d,
 \frac{b+d+\sqrt{(b-d)^2+4w^2}}{2},
 \frac{b+d-\sqrt{(b-d)^2+4w^2}}{2}
\right\}.
\label{eq:choi-spectrum}
\end{equation}
Consequently, complete positivity is equivalent to
\begin{equation}
a\ge z,
\qquad b\ge0,
\qquad d\ge0,
\qquad bd\ge w^2.
\label{eq:cp-general}
\end{equation}
Numerical evidence for indecomposable members in this sparse family
was originally reported in \cite{arx}.

Reference~\cite{CompanionSparseMaps} also shows that positivity of
\eqref{eq:map} is equivalent to non-negativity of a Hermitian
biquadratic form, or equivalently to positivity of a $3\times3$
Hermitian matrix $M(x)$ for all $x\neq0$. After passing to simplex
variables $u_i\ge0$ with $u_1+u_2+u_3=1$, the problem reduces to the
sign of a single scalar function.  This reduction is particularly
effective in the uniform bistochastic slice
\begin{equation}
a=b=d=\frac13\;,
\end{equation}
where the maps are both trace-preserving and unital, namely $\Tr
[\Phi [X]]=\Tr [X]$ and $\Phi (I)=I$. From now on we restrict to this
case, and we denote the maps by $\Phi(w,z)$ and their Choi matrices by
$W(w,z)$. We use the ordered product basis
\begin{equation}
e_1=|11\rangle,\ e_2=|12\rangle,\ e_3=|13\rangle,\ e_4=|21\rangle,\ e_5=|22\rangle,
\ e_6=|23\rangle,\ e_7=|31\rangle,\ e_8=|32\rangle,\ e_9=|33\rangle.
\end{equation}
We write $e_i e_j^*$ for the rank-one operator $|e_i\rangle\langle e_j|$.
In this basis the active off-diagonal Choi entries are precisely the
two coherence pairs $(e_2,e_4)$ and $(e_1,e_9)$, weighted by $w$ and
$z$, respectively. Using Eq. (\ref{eq:ChoiM}), we can then
write
\begin{eqnarray}
W(w,z)= \frac13 I_9+w (e_2e_4^*+e_4e_2^*)+ z(e_1e_9^*+e_9e_1^*)
=  \frac13 \left(
\begin{array}{ccc|ccc|ccc}
1 & \cdot & \cdot & \cdot & \cdot & \cdot & \cdot & \cdot & 3z \\
\cdot & 1 & \cdot & 3w & \cdot & \cdot & \cdot & \cdot & \cdot \\
\cdot & \cdot & 1 & \cdot & \cdot & \cdot & \cdot & \cdot & \cdot \\
\hline
\cdot & 3w & \cdot & 1 & \cdot & \cdot & \cdot & \cdot & \cdot \\
\cdot & \cdot & \cdot & \cdot & 1 & \cdot & \cdot & \cdot & \cdot \\
\cdot & \cdot & \cdot & \cdot & \cdot & 1 & \cdot & \cdot & \cdot \\
\hline
\cdot & \cdot & \cdot & \cdot & \cdot & \cdot & 1 & \cdot & \cdot \\
\cdot & \cdot & \cdot & \cdot & \cdot & \cdot & \cdot & 1 & \cdot \\
3z & \cdot & \cdot & \cdot & \cdot & \cdot & \cdot & \cdot & 1
\end{array}
\right),
\label{eq:W-wz}
\end{eqnarray}
where dots denote zero entries.

\subsection{A simple nonunital deformation}
\label{subsec:nonunital}
Although the exactly solved slice considered below is bistochastic, the sparse
construction is not tied to unitality. A simple nonunital positive map is
obtained by an invertible output filter. For $\lambda>0$, define
\begin{equation}
\Phi_{\lambda}(X)=D_{\lambda}\Phi(w,z)(X)D_{\lambda},
\qquad D_{\lambda}=\operatorname{diag}(1,1,\sqrt{\lambda}).
\label{eq:nonunital-filter}
\end{equation}
Since conjugation by $D_{\lambda}$ is completely positive,
$\Phi_{\lambda}$ is positive whenever $\Phi(w,z)$ is positive. Explicitly,
\begin{equation}
\Phi_{\lambda}(X)=
\begin{pmatrix}
 \Tr(X)/3 & wX_{21} & \sqrt{\lambda}\,zX_{13}\\
 wX_{12} & \Tr(X)/3 & 0\\
 \sqrt{\lambda}\,zX_{31} & 0 & \lambda\Tr(X)/3
\end{pmatrix},
\end{equation}
and $\Phi_{\lambda}(I)=\operatorname{diag}(1,1,\lambda)$, so the map is
nonunital for $\lambda\ne1$. If $D_{\lambda}$ is invertible,
indecomposability is also preserved: a decomposition of $\Phi_{\lambda}$ would,
after composition with the completely positive inverse filter
$X\mapsto D_{\lambda}^{-1}XD_{\lambda}^{-1}$, induce a decomposition of
$\Phi(w,z)$. This example is not trace preserving and is included to illustrate
how the present witnesses generate nonunital positive maps by local filtering.
More general nonuniform slices and their exact positivity thresholds are treated
in Ref.~\cite{CompanionSparseMaps}.

\section{Positivity region}
\label{sec:positivity}
Our first goal is to determine in the $(w,z)$-plane which maps $\Phi
(w,z)$ are positive. In the uniform slice the positivity problem can
be treated directly from the image of rank-one projectors.
\begin{lemma}[Direct scalar test in the uniform slice]
\label{lem:uniform-scalar-reduction}
The map $\Phi(w,z)$ is
positive if and only if
\begin{equation}
S(u_1,u_2):=\frac13-3w^2u_1u_2-3z^2u_1u_3\ge0,
\qquad
u_3:=1-u_1-u_2,
\label{eq:uniform-scalar-reduction}
\end{equation}
for all $u_1,u_2,u_3\ge0$ with $u_1+u_2+u_3=1$.
\end{lemma}

\begin{proof}
It is enough to test positivity on rank-one projectors.  Let
$y=(y_1,y_2,y_3)^T\neq0$, put
\begin{equation}
s=\|y_1\|^2+\|y_2\|^2+\|y_3\|^2,
\qquad
u_i=\frac{\|y_i\|^2}{s},
\end{equation}
so that $(u_1,u_2,u_3)$ belongs to the probability simplex.  From
\cref{eq:map}, with $a=b=d=1/3$, one obtains
\begin{equation}
\Phi(w,z)(\ket y\!\bra y)=
\begin{pmatrix}
 s/3 & w\,y_2\overline{y_1} & z\,y_1\overline{y_3}\\
 w\,y_1\overline{y_2} & s/3 & 0\\
 z\,\overline{y_1}y_3 & 0 & s/3
\end{pmatrix}.
\label{eq:uniform-rank-one-image}
\end{equation}
For $y\neq0$ all diagonal entries are strictly positive.  Taking the Schur
complement of the lower diagonal $2\times2$ block gives the single condition
\begin{equation}
\frac{s}{3}-\frac{3w^2\|y_1\|^2\|y_2\|^2}{s}
              -\frac{3z^2\|y_1\|^2\|y_3\|^2}{s}\ge0.
\end{equation}
Dividing by $s$ gives exactly \cref{eq:uniform-scalar-reduction}.  Conversely,
if \cref{eq:uniform-scalar-reduction} holds for every point of the simplex,
then \cref{eq:uniform-rank-one-image} is positive semidefinite for every
$y\neq0$, and therefore $\Phi(w,z)$ is positive.
\end{proof}

\begin{proposition}[Positivity conditions for $\Phi(w,z)$]
\label{prop:positivity-unbalanced}
Let $w,z\ge0$. Then $\Phi(w,z)$ is positive if and only if
\begin{equation}
0\le w\le\frac23,
\qquad
0\le z\le\frac23.
\label{eq:positivity-square}
\end{equation}
Thus the positivity region in the $(w,z)$-plane is the square
$[0,2/3]\times[0,2/3]$.
\end{proposition}

\begin{proof}
By \cref{lem:uniform-scalar-reduction}, positivity is equivalent to
non-negativity of $S(u_1,u_2)$ on the simplex.  Fix $u_1=t\in[0,1]$.  For $t=0$ the condition is automatic.
For $t>0$, the variable $u_2$ ranges in $[0,1-t]$, and
\begin{equation}
S(t,u_2)=\frac13-3t\left[w^2u_2+z^2(1-t-u_2)\right]
\end{equation}
is affine in $u_2$.  Hence its minimum on this interval is attained at one of
the endpoints.  If $w\ge z$, the minimum is at $u_2=1-t$; if $z\ge w$, it is at
$u_2=0$.  Therefore
\begin{equation}
\min_{0\le u_2\le1-t}S(t,u_2)
=\frac13-3t(1-t)\max\{w^2,z^2\}.
\label{eq:S-min}
\end{equation}
Since $t(1-t)\le1/4$, with equality at $t=1/2$, positivity on the whole simplex
is equivalent to
\begin{equation}
\frac13-\frac34\max\{w^2,z^2\}\ge0.
\end{equation}
This is equivalent to \cref{eq:positivity-square}.
\end{proof}

\begin{corollary}
\label{cor:cp-square}
The completely positive region for $\Phi(w,z)$ is
$[0,1/3]\times[0,1/3]$.
\end{corollary}

\begin{proof}
The Choi spectrum in \cref{eq:choi-spectrum}, specialized to
$a=b=d=1/3$, gives complete positivity if and only if
$z\le1/3$ and $w \le1/3$. 
\end{proof}

\begin{figure}[t]
 \centering
\includegraphics[width=0.46\textwidth]{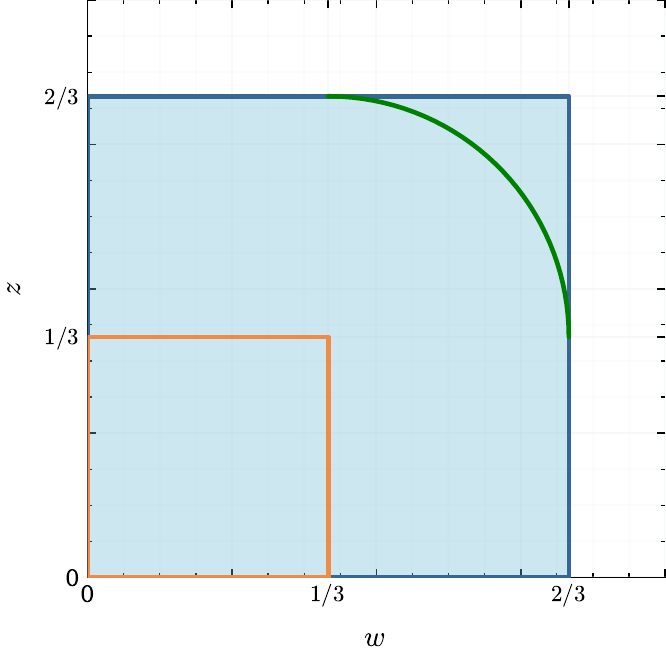}
 \caption{Phase diagram of the map $\Phi(w,z)$ in the
$(w,z)$-plane. The blue square is the positivity region
$[0,2/3]\times[0,2/3]$ from \cref{prop:positivity-unbalanced}, and the orange
square is the completely positive region $[0,1/3]\times[0,1/3]$ from
\cref{cor:cp-square}. In the upper-right corner, the green quarter circle
\cref{eq:quarter-circle} separates the decomposable sector from the
indecomposable cap proved in \cref{prop:decomposability-unbalanced}.}
 \label{fig:unbalanced-geometry}
\end{figure}

Thus, positivity is controlled only by the larger of the
two active coherence weights.  The smaller complete-positivity square is the
region where the two active Choi blocks are already positive semidefinite.
The remaining question is decomposability inside the positivity square.

\section{Quantum-channel and entanglement-breaking regimes}
\label{sec:channels}
Inside the completely positive square, $\Phi(w,z)$ is a trace-preserving
completely positive map and hence a qutrit quantum channel. At the origin it is
the completely depolarizing channel,
\begin{equation}
\Phi(0,0)(X)=\frac{\Tr X}{3}I_3.
\end{equation}
The parameters $w$ and $z$ retain, respectively, the input coherences $X_{12}$
and $X_{13}$, while the $23$ coherence is completely erased. Thus the CP square
is an anisotropic two-coherence deformation of the completely depolarizing
channel. The following result gives its operational classification.

\begin{proposition}[Entanglement-breaking region]
\label{prop:entanglement-breaking}
In the uniform bistochastic slice, $\Phi(w,z)$ is entanglement breaking if and
only if it is completely positive. Equivalently,
\begin{equation}
\Phi(w,z)\text{ is entanglement breaking}
\quad\Longleftrightarrow\quad
0\le w\le\frac13,\qquad 0\le z\le\frac13.
\label{eq:eb-square}
\end{equation}
\end{proposition}

\begin{proof}
An entanglement-breaking map is necessarily completely positive, so it remains
to prove that $W(w,z)$ is separable throughout the CP square
\cite{HorodeckiShorRuskai2003}. Put $p=3w$ and $q=3z$. By affine dependence,
\begin{align}
W(w,z)={}&(1-p)(1-q)W(0,0)
+p(1-q)W\!\left(\frac13,0\right)\nonumber\\
&+(1-p)qW\!\left(0,\frac13\right)
+pqW\!\left(\frac13,\frac13\right).
\label{eq:eb-convex}
\end{align}
It is therefore enough to prove separability at the four vertices.
Clearly, $W(0,0)=I_9/3$ is separable. For the three nontrivial vertices, define
\begin{align}
x_w(\theta)
&=(1,e^{2i\theta},e^{-3i\theta})^T,
&
y_w(\theta)
&=\frac1{\sqrt3}(1,e^{2i\theta},e^{i\theta})^T,
\\
x_z(\theta)
&=(1,e^{i\theta},e^{2i\theta})^T,
&
y_z(\theta)
&=\frac1{\sqrt3}(1,e^{3i\theta},e^{-2i\theta})^T,
\\
x_{wz}(\theta)
&=(1,e^{i\theta},e^{-3i\theta})^T,
&
y_{wz}(\theta)
&=\frac1{\sqrt3}(1,e^{i\theta},e^{3i\theta})^T.
\end{align}
Direct phase averaging then gives
\begin{align}
W\left(\frac13,0\right)
&=
\int_0^{2\pi}\frac{d\theta}{2\pi}
\ket{x_w(\theta)\otimes y_w(\theta)}
\bra{x_w(\theta)\otimes y_w(\theta)},
\\
W\left(0,\frac13\right)
&=
\int_0^{2\pi}\frac{d\theta}{2\pi}
\ket{x_z(\theta)\otimes y_z(\theta)}
\bra{x_z(\theta)\otimes y_z(\theta)},
\\
W\left(\frac13,\frac13\right)
&=
\int_0^{2\pi}\frac{d\theta}{2\pi}
\ket{x_{wz}(\theta)\otimes y_{wz}(\theta)}
\bra{x_{wz}(\theta)\otimes y_{wz}(\theta)}.
\end{align}
Indeed, in the first integral the only coincident phase exponents correspond to
the pair $(e_2,e_4)$; in the second, they correspond to $(e_1,e_9)$; and in the
third, they correspond to both pairs. All other off-diagonal terms vanish under
phase averaging. Hence all four vertices are separable. Equation~\eqref{eq:eb-convex}
and convexity of the separable cone prove the assertion.
\end{proof}
Consequently, the family contains no CP but non-entanglement-breaking channel:
the complete-positivity and entanglement-breaking boundaries coincide. Each
channel in the CP square therefore admits a measure-and-prepare realization,
while the positive maps outside that square are not physical channels by
themselves but define entanglement witnesses through their Choi matrices
\cite{HorodeckiShorRuskai2003}. This provides a direct quantum-information
interpretation of all three principal regions in \cref{fig:unbalanced-geometry}.

\section{Decomposability boundary}
\label{sec:decomposability}

We now determine exactly which positive maps in \cref{eq:positivity-square} are
decomposable.  The answer fills the positivity square except for a circular cap
in the upper-right corner.  The proof is split into two parts.  First, explicit
PPT states detect all points beyond the circle, proving indecomposability.
Second, an explicit representation $W(w,z)=P+Q^\tau$ proves decomposability on and
below the circle.

\begin{proposition}[Decomposability conditions for $\Phi(w,z)$]
\label{prop:decomposability-unbalanced}
Let $w,z\geq 0$. Then $\Phi (w,z)$ is decomposable if and only if
\begin{equation}
\bigl(\max\{w-\tfrac13,0\}\bigr)^2+
\bigl(\max\{z-\tfrac13,0\}\bigr)^2\le\frac19.
\label{eq:decomp-region}
\end{equation}
Equivalently, every positive map with either $w\le1/3$ or $z\le1/3$ is
decomposable, while in the corner $w,z\ge1/3$ the boundary between
decomposable and indecomposable maps is
\begin{equation}
\left(w-\frac13\right)^2+
\left(z-\frac13\right)^2=\frac19.
\label{eq:quarter-circle}
\end{equation}
\end{proposition}

\begin{proof}
We first prove indecomposability outside the region
\cref{eq:decomp-region}.  By \cref{thm:ppt-detects-indecomp}, it is enough to
construct a PPT state with negative expectation value on $W(w,z)$.

Let $\phi\in[0,\pi/2]$, set $s=\sin\phi$ and $c=\cos\phi$, and define
\begin{equation}
  \rho_{\phi}=
  \frac{1}{2(1+s+c)}
\left(
\begin{array}{ccc|ccc|ccc}
1 & \cdot & \cdot & \cdot & \cdot & \cdot & \cdot & \cdot & -s \\
\cdot & c & \cdot & -c & \cdot & \cdot & \cdot & \cdot & \cdot \\
\cdot & \cdot & s & \cdot & \cdot & \cdot & \cdot & \cdot & \cdot \\
\hline
\cdot & -c & \cdot & c & \cdot & \cdot & \cdot & \cdot & \cdot \\
\cdot & \cdot & \cdot & \cdot & c^2 & \cdot & \cdot & \cdot & \cdot \\
\cdot & \cdot & \cdot & \cdot & \cdot & \cdot & \cdot & \cdot & \cdot \\
\hline
\cdot & \cdot & \cdot & \cdot & \cdot & \cdot & s & \cdot & \cdot \\
\cdot & \cdot & \cdot & \cdot & \cdot & \cdot & \cdot & \cdot & \cdot \\
-s & \cdot & \cdot & \cdot & \cdot & \cdot & \cdot & \cdot & s^2
\end{array}
\right).
\label{eq:rho-uv}
\end{equation}
The matrix $\rho_\phi$ has unit trace and is positive semidefinite because it is a direct sum of
positive rank-one $2\times2$ blocks and nonnegative scalar blocks.  The same is
true after partial transposition: the coherences supported on $(e_2,e_4)$ and
$(e_1,e_9)$ move to the pairs $(e_1,e_5)$ and $(e_3,e_7)$, respectively, again
producing only positive $2\times2$ rank-one blocks and nonnegative scalars.
Thus $\rho_\phi$ is PPT.

Using the explicit witness form \cref{eq:W-wz}, one obtains
\begin{equation}
\Tr\bigl[W(w,z)\rho_{\phi}\bigr]
=\frac{\frac13(1+s+c)-zs-wc}{1+s+c}.
\label{eq:witness-action}
\end{equation}
Assume now that $w,z\ge1/3$ and that \cref{eq:decomp-region} is violated.  Put
\begin{equation}
r:=\sqrt{\left(z-\frac13\right)^2+\left(w-\frac13\right)^2},
\qquad
s:=\frac{z-1/3}{r},
\qquad
c:=\frac{w-1/3}{r}.
\label{eq:uv-param}
\end{equation}
Then $s,c\ge0$, $s^2+c^2=1$, and $r>1/3$.  Substitution in
\cref{eq:witness-action} gives
\begin{equation}
\Tr\bigl[W(w,z)\rho_{\phi}\bigr]
=\frac{\frac13-r}{1+s+c}<0.
\end{equation}
Hence $\Phi(w,z)$ is indecomposable outside \cref{eq:decomp-region}.

It remains to prove decomposability inside \cref{eq:decomp-region}.  Let
\begin{equation}
t:=\frac13,
\qquad
\eta:=(z-t)_+,
\qquad
\xi:=\min\{w,t\},
\qquad
\gamma:=\min\{z,t\},
\qquad
\delta:=(w-t)_+,
\label{eq:etaxi-split}
\end{equation}
where $(x)_+=\max\{x,0\}$.  Then $z=\eta+\gamma$ and $w=\xi+\delta$, with all
four parameters between $0$ and $t$.  The condition
\cref{eq:decomp-region} is equivalent to
\begin{equation}
\delta^2\le t^2-\eta^2.
\label{eq:aux-bound}
\end{equation}
We now write
\begin{equation}
P=
\left(
\begin{array}{ccc|ccc|ccc}
\eta^2/t & \cdot & \cdot & \cdot & \cdot & \cdot & \cdot & \cdot & \eta \\
\cdot & \xi & \cdot & \xi & \cdot & \cdot & \cdot & \cdot & \cdot \\
\cdot & \cdot & t-\gamma & \cdot & \cdot & \cdot & \cdot & \cdot & \cdot \\
\hline
\cdot & \xi & \cdot & \xi & \cdot & \cdot & \cdot & \cdot & \cdot \\
\cdot & \cdot & \cdot & \cdot & \cdot & \cdot & \cdot & \cdot & \cdot \\
\cdot & \cdot & \cdot & \cdot & \cdot & t & \cdot & \cdot & \cdot \\
\hline
\cdot & \cdot & \cdot & \cdot & \cdot & \cdot & t-\gamma & \cdot & \cdot \\
\cdot & \cdot & \cdot & \cdot & \cdot & \cdot & \cdot & t & \cdot \\
\eta & \cdot & \cdot & \cdot & \cdot & \cdot & \cdot & \cdot & t
\end{array}
\right),
\label{eq:P-matrix}
\end{equation}
and
\begin{equation}
Q=
\left(
\begin{array}{ccc|ccc|ccc}
t-\eta^2/t & \cdot & \cdot & \cdot & \delta & \cdot & \cdot & \cdot & \cdot \\
\cdot & t-\xi & \cdot & \cdot & \cdot & \cdot & \cdot & \cdot & \cdot \\
\cdot & \cdot & \gamma & \cdot & \cdot & \cdot & \gamma & \cdot & \cdot \\
\hline
\cdot & \cdot & \cdot & t-\xi & \cdot & \cdot & \cdot & \cdot & \cdot \\
\delta & \cdot & \cdot & \cdot & t & \cdot & \cdot & \cdot & \cdot \\
\cdot & \cdot & \cdot & \cdot & \cdot & \cdot & \cdot & \cdot & \cdot \\
\hline
\cdot & \cdot & \gamma & \cdot & \cdot & \cdot & \gamma & \cdot & \cdot \\
\cdot & \cdot & \cdot & \cdot & \cdot & \cdot & \cdot & \cdot & \cdot \\
\cdot & \cdot & \cdot & \cdot & \cdot & \cdot & \cdot & \cdot & \cdot
\end{array}
\right).
\label{eq:Q-matrix}
\end{equation}
A direct comparison of the nonzero entries shows that indeed $W(w,z)=P+Q^\tau$.

The matrix $P$ is positive semidefinite because it is the direct sum of the
positive blocks
\begin{equation}
\begin{pmatrix}\eta^2/t & \eta\\ \eta & t\end{pmatrix},
\qquad
\xi\begin{pmatrix}1 & 1\\ 1 & 1\end{pmatrix},
\qquad
t-\gamma,
\qquad t,
\qquad t-\gamma,
\qquad t.
\end{equation}
The matrix $Q$ is positive semidefinite because it is the direct sum of
\begin{equation}
\begin{pmatrix}t-\eta^2/t & \delta\\ \delta & t\end{pmatrix},
\qquad
\gamma\begin{pmatrix}1 & 1\\ 1 & 1\end{pmatrix},
\qquad
t-\xi,
\qquad t-\xi,
\end{equation}
together with zero blocks.  The first $2\times2$ block is positive exactly by
\cref{eq:aux-bound}; all other displayed blocks are manifestly positive.
Therefore $W(w,z)$ is decomposable whenever \cref{eq:decomp-region} holds.
Combining this with the PPT-state detection above proves the proposition.
\end{proof}
\section{Witness-adapted PPT edge-state families}
\label{sec:pptstates}
The PPT entangled states $\rho_{\phi}$ for $\phi \in (0,\pi/2)$ of
\cref{eq:rho-uv} used in the proof of
\cref{prop:decomposability-unbalanced} are adapted to the two active
coherences of the sparse witness.  In this section we turn this
observation into an explicit family of PPT edge states detected by the
endpoint witness
\begin{equation}
W_*:=W\!\left(\frac23,\frac23\right)
=
\frac13 I_9+\frac23(e_2e_4^*+e_4e_2^*)
+
\frac23(e_1e_9^*+e_9e_1^*) .
\end{equation}
The construction is kernel-first: we prescribe compatible kernels for a state
and its partial transpose, then solve inside the corresponding PPT face.  The
edge property and detection region then remain explicit.

We use the following form of the range criterion~\cite{Horodecki1997}.
\begin{theorem}[Range criterion]
Let $\rho$ be a nonzero separable positive semidefinite operator on
$\mathbb C^m\otimes\mathbb C^n$. Then there exists a family of product
vectors $\{x_i\otimes y_i\}_i$ such that
\begin{eqnarray}
\Ran\rho=\spanop\{x_i\otimes y_i\}_i,
\qquad
\Ran\rho^\tau=\spanop\{x_i\otimes \overline{y_i}\}_i .
\end{eqnarray}
In particular, there exists a nonzero product vector
$x\otimes y\in\Ran\rho$ such that
$x\otimes\overline y\in\Ran\rho^\tau$.
\end{theorem}

Let us also recall the edge-state terminology.

\noindent\textbf{Definition (edge state).} A PPT entangled state $\rho$ is called
an \emph{edge state} if there is no nonzero product vector
$x\otimes y\in\Ran\rho$ such that
$x\otimes\overline y\in\Ran\rho^\tau$. Equivalently, one cannot subtract any
nonzero projector onto a product vector while preserving both positivity and the
PPT property~\cite{Lewenstein2000,HorodeckiLewensteinVidalCirac2000}.

\subsection{PPT edge states detected by the endpoint witness}
For parameters $\alpha,\beta,\mu,\nu>0$, define
\begin{equation}
A_{\alpha,\beta,\mu,\nu}
:=
(e_1-\alpha e_9)(e_1-\alpha e_9)^*
+\frac{\beta}{\mu}(\mu e_2-e_4)(\mu e_2-e_4)^*
+\frac{\alpha}{\nu}e_3e_3^*
+\beta^2e_5e_5^*
+\alpha\nu\,e_7e_7^*,
\end{equation}
and set
\begin{equation}
\rho_{\alpha,\beta,\mu,\nu}:=
\frac{A_{\alpha,\beta,\mu,\nu}}{N_{\alpha,\beta,\mu,\nu}},
\qquad
N_{\alpha,\beta,\mu,\nu}:=\Tr A_{\alpha,\beta,\mu,\nu}.
\end{equation}
Explicitly,
\begin{equation}
N_{\alpha,\beta,\mu,\nu}
=1+\alpha^2+
\beta\left(\mu+\frac1\mu\right)
+
\alpha\left(\nu+\frac1\nu\right)
+
\beta^2 .
\end{equation}
In the ordered basis fixed in Sec.~\ref{sec:family}, this state has the matrix
form
\begin{equation}
\rho_{\alpha,\beta,\mu,\nu}
=
\frac1{N_{\alpha,\beta,\mu,\nu}}
\left(
\begin{array}{ccc|ccc|ccc}
1 & \cdot & \cdot & \cdot & \cdot & \cdot & \cdot & \cdot & -\alpha\\
\cdot & \beta\mu & \cdot & -\beta & \cdot & \cdot & \cdot & \cdot & \cdot\\
\cdot & \cdot & \alpha/\nu & \cdot & \cdot & \cdot & \cdot & \cdot & \cdot\\
\hline
\cdot & -\beta & \cdot & \beta/\mu & \cdot & \cdot & \cdot & \cdot & \cdot\\
\cdot & \cdot & \cdot & \cdot & \beta^2 & \cdot & \cdot & \cdot & \cdot\\
\cdot & \cdot & \cdot & \cdot & \cdot & \cdot & \cdot & \cdot & \cdot\\
\hline
\cdot & \cdot & \cdot & \cdot & \cdot & \cdot & \alpha\nu & \cdot & \cdot\\
\cdot & \cdot & \cdot & \cdot & \cdot & \cdot & \cdot & \cdot & \cdot\\
-\alpha & \cdot & \cdot & \cdot & \cdot & \cdot & \cdot & \cdot & \alpha^2
\end{array}
\right).
\end{equation}
The states $\rho_\phi$ from \cref{eq:rho-uv} are recovered by taking
$\alpha=\sin\phi$, $\beta=\cos\phi$, and $\mu=\nu=1$.

\begin{proposition}
For every $\alpha,\beta,\mu,\nu>0$, the state
$\rho_{\alpha,\beta,\mu,\nu}$ is a PPT edge state of rank type $(5,5)$. Moreover,
\begin{equation}
\Tr(W_*\rho_{\alpha,\beta,\mu,\nu})
=
\frac{(\alpha-1)^2+(\beta-1)^2-1
+\alpha(\nu+\nu^{-1}-2)+\beta(\mu+\mu^{-1}-2)}
{3N_{\alpha,\beta,\mu,\nu}} .
\end{equation}
Hence $W_*$ detects $\rho_{\alpha,\beta,\mu,\nu}$ if and only if
\begin{equation}
(\alpha-1)^2+(\beta-1)^2
+\alpha(\nu+\nu^{-1}-2)+\beta(\mu+\mu^{-1}-2)<1 .
\label{condd}
\end{equation}
\end{proposition}

\begin{proof}
Positivity is immediate from the displayed rank-one decomposition of
$A_{\alpha,\beta,\mu,\nu}$. Under partial transposition on the second subsystem,
\begin{equation}
e_2e_4^*+e_4e_2^* \longmapsto e_1e_5^*+e_5e_1^*,
\qquad
e_1e_9^*+e_9e_1^* \longmapsto e_3e_7^*+e_7e_3^* .
\end{equation}
Thus
\begin{align}
A_{\alpha,\beta,\mu,\nu}^\tau
&=(e_1-\beta e_5)(e_1-\beta e_5)^*+
\beta\mu\,e_2e_2^*
+\frac{\alpha}{\nu}(e_3-\nu e_7)(e_3-\nu e_7)^* \\
&\quad
+\frac{\beta}{\mu}e_4e_4^*
+\alpha^2 e_9e_9^*\succeq0 .
\end{align}
Therefore $\rho_{\alpha,\beta,\mu,\nu}$ is PPT. The two rank-one decompositions
also give
\begin{equation}
\ker\rho_{\alpha,\beta,\mu,\nu}
=
\operatorname{span}\{
\alpha e_1+e_9,
\ e_2+\mu e_4,
\ e_6,
\ e_8
\},
\end{equation}
and
\begin{equation}
\ker\rho_{\alpha,\beta,\mu,\nu}^\tau
=
\operatorname{span}\{
\beta e_1+e_5,
\ \nu e_3+e_7,
\ e_6,
\ e_8
\}.
\end{equation}
Consequently,
\begin{equation}
\operatorname{rank}\rho_{\alpha,\beta,\mu,\nu}
=
\operatorname{rank}\rho_{\alpha,\beta,\mu,\nu}^\tau=5.
\end{equation}

It remains to verify the edge property. Suppose that a product vector
$x\otimes y$ satisfies
\begin{equation}
x\otimes y\in\Ran\rho_{\alpha,\beta,\mu,\nu},
\qquad
x\otimes\overline y\in\Ran\rho_{\alpha,\beta,\mu,\nu}^\tau .
\end{equation}
Write $x=(x_1,x_2,x_3)^T$ and $y=(y_1,y_2,y_3)^T$. The first range condition is
equivalent to orthogonality to $\ker\rho_{\alpha,\beta,\mu,\nu}$, namely
\begin{equation}
x_2y_3=0,
\qquad
x_3y_2=0,
\qquad
x_1y_2+\mu x_2y_1=0,
\qquad
x_3y_3+\alpha x_1y_1=0.
\end{equation}
The second range condition gives
\begin{equation}
x_2\overline{y_3}=0,
\qquad
x_3\overline{y_2}=0,
\qquad
x_2\overline{y_2}+\beta x_1\overline{y_1}=0,
\qquad
x_3\overline{y_1}+\nu x_1\overline{y_3}=0.
\end{equation}
These equations have no nonzero product solution. Indeed, if $y_2\neq0$, then
$x_3=0$, whence $x_1y_1=0$ from
$x_3y_3+\alpha x_1y_1=0$. The equation
$x_2\overline{y_2}+\beta x_1\overline{y_1}=0$ then gives $x_2=0$, and finally
$x_1y_2+\mu x_2y_1=0$ gives $x_1=0$, a contradiction. Hence $y_2=0$. Then
$x_2\overline{y_2}+\beta x_1\overline{y_1}=0$ implies $x_1y_1=0$. If
$y_1\neq0$, then $x_1=0$; the equations
$x_1y_2+\mu x_2y_1=0$ and
$x_3\overline{y_1}+\nu x_1\overline{y_3}=0$ then give $x_2=x_3=0$, again a
contradiction. Thus $y_1=0$. With $y_1=y_2=0$, either $y_3\neq0$, in which case
$x_2=x_1=0$ and then $x_3=0$, or $y=0$. In all cases the product vector is zero.
Therefore no nonzero product vector satisfies the range criterion, and
$\rho_{\alpha,\beta,\mu,\nu}$ is an edge state.

Finally, the pairing with the endpoint witness is
\begin{equation}
\Tr(W_*\rho_{\alpha,\beta,\mu,\nu})
=
\frac13+ \frac43 (\rho_{\alpha,\beta,\mu,\nu})_{24}
+\frac43(\rho_{\alpha,\beta,\mu,\nu})_{19}
=
\frac13-\frac{4(\alpha+\beta)}{3N_{\alpha,\beta,\mu,\nu}} .
\end{equation}
Since
\begin{equation}
N_{\alpha,\beta,\mu,\nu}-4(\alpha+\beta)
=
(\alpha-1)^2+(\beta-1)^2-1
+\alpha(\nu+\nu^{-1}-2)+\beta(\mu+\mu^{-1}-2),
\end{equation}
the stated detection formula and condition follow.
\end{proof}
The detection condition \cref{condd} is an explicit analytic
description of the region of this four-parameter edge-state family
detected by the endpoint witness. The parameters $\alpha$ and $\beta$
control the two coherences paired with $W_*$, while $\mu$ and $\nu$
deform the range/kernel geometry while preserving both the PPT
property and the edge property. This is consistent with the general
range-based approach to $3\otimes3$ PPT edge states and Choi-type
witnesses~\cite{HaKye2005,HaKye2012,Clarisse2006}.

\begin{corollary}[Symmetric two-parameter slice]
Setting $\mu=\nu=1$ gives the two-parameter family
$\rho_{\alpha,\beta}:=\rho_{\alpha,\beta,1,1}$. For every
$\alpha,\beta>0$, this state is a PPT edge state of rank type $(5,5)$, and
\begin{equation}
\Tr(W_*\rho_{\alpha,\beta})
=
\frac{(\alpha-1)^2+(\beta-1)^2-1}{3N_{\alpha,\beta}},
\qquad
N_{\alpha,\beta}:=N_{\alpha,\beta,1,1}=(1+\alpha)^2+\beta(\beta+2).
\end{equation}
Thus $W_*$ detects $\rho_{\alpha,\beta}$ precisely on the disk
\begin{equation}
(\alpha-1)^2+(\beta-1)^2<1.
\end{equation}
Moreover, since
\begin{equation}
\mu+\mu^{-1}\ge2,
\qquad
\nu+\nu^{-1}\ge2,
\end{equation}
with equality if and only if $\mu=\nu=1$, the symmetric slice is the most
favorable one for detection at fixed $(\alpha,\beta)$.
\end{corollary}

For fixed $(\mu,\nu)$, the detected region is a shifted disk in the
$(\alpha,\beta)$-plane, with maximal size at $\mu=\nu=1$; see
\cref{fig:detection-region-four-parameter}.  For fixed
$(\alpha,\beta)$, the detected set in the $(\mu,\nu)$-plane is a
plectrum-shaped neighborhood of $(1,1)$ whenever \cref{condd} is
satisfiable; see \cref{fig:detection-region-mu-nu}.

\begin{figure}[t]
 \centering
  \includegraphics[width=.6\textwidth]{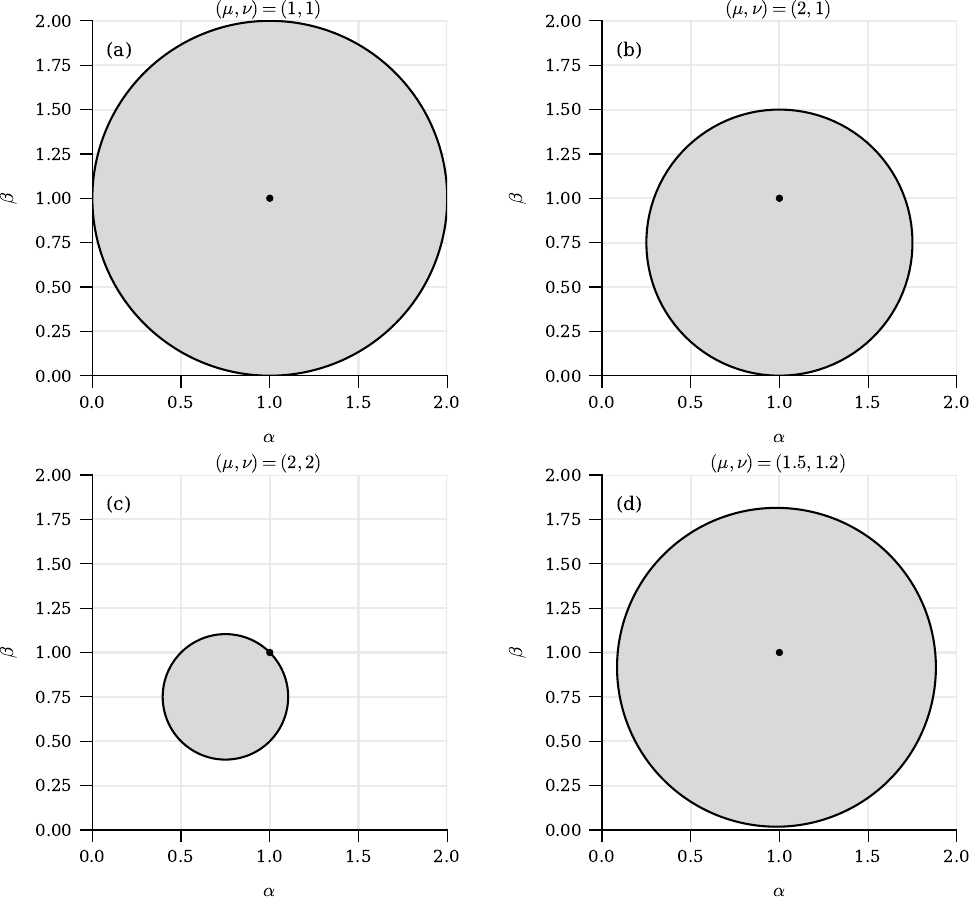}
 \caption{Detection region in the $(\alpha,\beta)$-plane for the
 four-parameter family $\rho_{\alpha,\beta,\mu,\nu}$ for several fixed values
 of $(\mu,\nu)$. The shaded region consists of those parameters for which the
 endpoint witness $W_*=W(\frac23,\frac23)$ detects the state, equivalently
 $(\alpha-1)^2+(\beta-1)^2+\alpha(\nu+\nu^{-1}-2)+\beta(\mu+\mu^{-1}-2)<1$.
 The black point marks $(\alpha,\beta)=(1,1)$. The top-left panel,
 corresponding to $(\mu,\nu)=(1,1)$, is the symmetric slice, where the
 detection region is the disk $(\alpha-1)^2+(\beta-1)^2<1$. As $\mu$ and/or
 $\nu$ move away from $1$, the region shrinks, in agreement with
 $\mu+\mu^{-1}\ge2$ and $\nu+\nu^{-1}\ge2$.}
 \label{fig:detection-region-four-parameter}
\end{figure}

\begin{figure}[t]
 \centering
 \includegraphics[width=.6\textwidth]{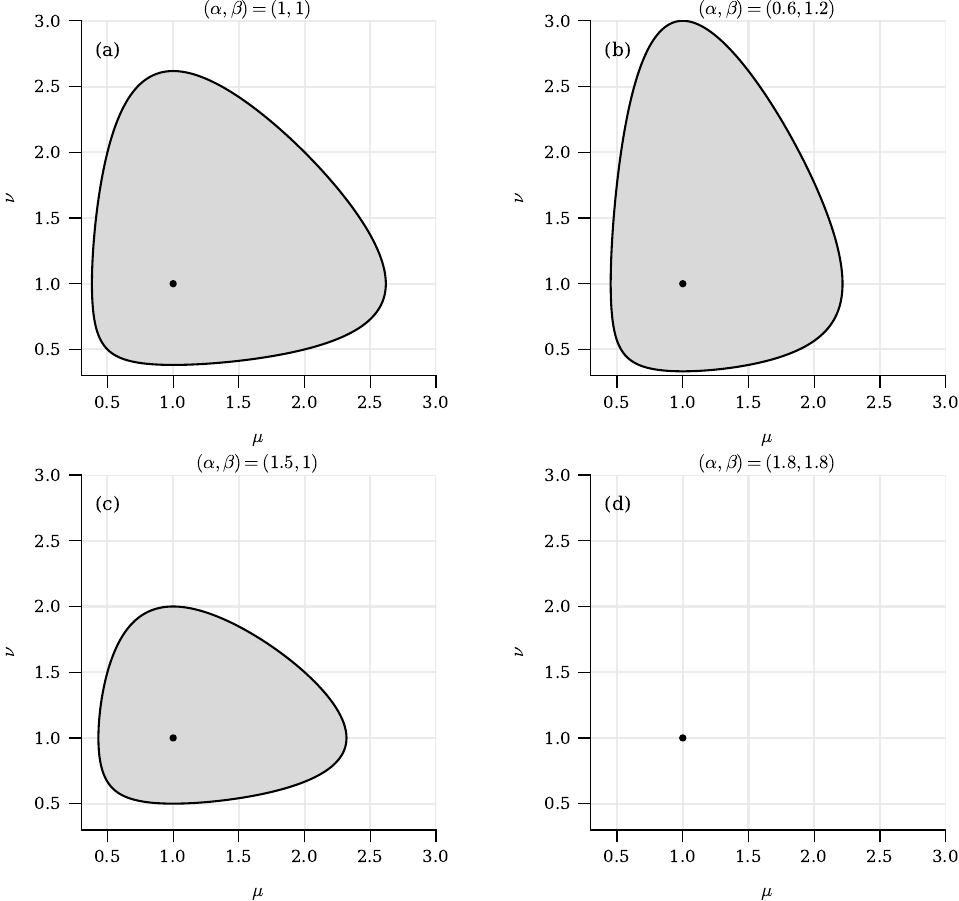}
 \caption{Detection region of the endpoint witness $W_*=W(\frac23,\frac23)$
 in the $(\mu,\nu)$-plane for the four-parameter family
 $\rho_{\alpha,\beta,\mu,\nu}$ at several fixed values of $(\alpha,\beta)$. The
 shaded region consists of those parameters for which
 $(\alpha-1)^2+(\beta-1)^2+\alpha(\nu+\nu^{-1}-2)+\beta(\mu+\mu^{-1}-2)<1$.
 The black point marks $(\mu,\nu)=(1,1)$. Since $\mu+\mu^{-1}\ge2$ and
 $\nu+\nu^{-1}\ge2$, with equality only at $\mu=\nu=1$, the point $(1,1)$ is
 always the most favorable one for detection at fixed $(\alpha,\beta)$. The
 bottom-right panel illustrates a parameter choice for which the detection region
 is empty.}
 \label{fig:detection-region-mu-nu}
\end{figure}

\subsection{Exposedness of the constructed PPT edge states}
We next show that these PPT edge states generate exposed rays of the PPT cone.

\noindent\textbf{Definition (exposed ray of the PPT cone).}
Let
\begin{equation}
T=\{X\in M_3\otimes M_3:\ X\succeq0,\ X^\tau\succeq0\}
\end{equation}
be the PPT cone of possibly unnormalized PPT operators. A ray $\mathbb R_+X_0\subset T$ is an
\emph{exposed ray} of $T$ if there exists a real linear functional $F$,
nonnegative on $T$, such that
\begin{equation}
T\cap\{X:F(X)=0\}=\mathbb R_+X_0.
\end{equation}
Equivalently, $\mathbb R_+X_0$ is a one-dimensional exposed face of $T$~\cite{Kye2011Faces,FacialStructures2012}.

We can then prove the following.

\begin{proposition}[Exposedness of the four-parameter family]
\label{cor:four-parameter-exposed}
For every $\alpha,\beta,\mu,\nu>0$, the state
$\rho_{\alpha,\beta,\mu,\nu}$ generates an exposed ray of the PPT cone $T$.
Equivalently, $\rho_{\alpha,\beta,\mu,\nu}$ is an exposed PPT edge state.
\end{proposition}

\begin{proof}
Let
\begin{equation}
D:=\Ran\rho_{\alpha,\beta,\mu,\nu},
\qquad
E:=\Ran\rho_{\alpha,\beta,\mu,\nu}^{\tau}.
\end{equation}
From the rank-one decompositions of
$A_{\alpha,\beta,\mu,\nu}$ and
$A_{\alpha,\beta,\mu,\nu}^{\tau}$ one has
\begin{equation}
D^\perp=
\operatorname{span}\{
\alpha e_1+e_9,\,
e_2+\mu e_4,\,
e_6,\,
e_8
\},
\end{equation}
and
\begin{equation}
E^\perp=
\operatorname{span}\{
\beta e_1+e_5,\,
\nu e_3+e_7,\,
e_6,\,
e_8
\}.
\end{equation}
Consider the face of the PPT cone determined by the range pair $(D,E)$,
\begin{equation}
\mathcal F(D,E)
=
\{
X\in T:\ \Ran X\subset D,\ \Ran X^\tau\subset E
\}.
\end{equation}
Clearly
$\rho_{\alpha,\beta,\mu,\nu}\in\mathcal F(D,E)$. We now show that this face is
exactly the ray generated by this state.

Since $X=X^\dagger$ and $\Ran X\subset D$, equivalently
$D^\perp\subset\ker X$, every Hermitian matrix satisfying the first range
constraint can be written as
\begin{equation}
X=\sum_{i,j=1}^5 h_{ij} f_i f_j^*,
\qquad
H=(h_{ij})=H^\dagger,
\end{equation}
where
\begin{equation}
f_1=e_1-\alpha e_9,\quad
f_2=\mu e_2-e_4,\quad
f_3=e_3,\quad
f_4=e_5,\quad
f_5=e_7
\end{equation}
is a basis of $D$.

We impose the second range constraint,
$\Ran X^\tau\subset E$, equivalently
$E^\perp\subset\ker X^\tau$. Thus
\begin{equation}
X^\tau(\beta e_1+e_5)=0,\qquad
X^\tau(\nu e_3+e_7)=0,\qquad
X^\tau e_6=0,\qquad
X^\tau e_8=0.
\label{eq:four-kernel-constraints}
\end{equation}
The following calculation is purely linear-algebraic: the two range constraints give a homogeneous linear system for the coefficients $h_{ij}$.
We solve these equations using the rule
\begin{eqnarray}
(\ket{i,j}\bra{k,l})^\tau
=
\ket{i,l}\bra{k,j}.
\end{eqnarray}
From direct computation one obtains
\begin{equation}
X^\tau e_6
=
-\overline{h_{23}}\,e_1
+
h_{34}\,e_2
+
\alpha h_{12}\,e_7
-
\alpha h_{14}\,e_8.
\end{equation}
Hence
\begin{equation}
h_{12}=h_{14}=h_{23}=h_{34}=0.
\label{eq:first-offdiag}
\end{equation}
Similarly,
\begin{equation}
X^\tau e_8
=
\mu h_{25}\,e_1
-
\alpha\mu\,\overline{h_{12}}\,e_3
+
h_{45}\,e_4
-
\alpha\overline{h_{14}}\,e_6.
\end{equation}
Using \eqref{eq:first-offdiag}, this gives
\begin{equation}
h_{25}=h_{45}=0.
\label{eq:second-offdiag}
\end{equation}

Next impose $X^\tau(\beta e_1+e_5)=0$. After using
\eqref{eq:first-offdiag} and \eqref{eq:second-offdiag}, the only remaining
terms are
\begin{align}
X^\tau(\beta e_1+e_5)
&=
(\beta h_{11}-\mu h_{22})e_1
+\mu h_{24}e_2
+\beta h_{13}e_3
-\overline{h_{24}}e_4
\nonumber\\
&\quad
+(-\beta\mu h_{22}+h_{44})e_5
+\beta\overline{h_{15}}e_7
+\beta\overline{h_{35}}e_9 .
\end{align}
Therefore
\begin{equation}
h_{13}=h_{15}=h_{24}=h_{35}=0,
\end{equation}
and
\begin{equation}
h_{22}=\frac{\beta}{\mu}h_{11},
\qquad
h_{44}=\beta^2 h_{11}.
\label{eq:diag-beta}
\end{equation}

Finally, impose $X^\tau(\nu e_3+e_7)=0$. With all off-diagonal coefficients
already eliminated, this reduces to
\begin{equation}
X^\tau(\nu e_3+e_7)
=
(-\alpha h_{11}+\nu h_{33})e_3
+
(-\alpha\nu h_{11}+h_{55})e_7.
\end{equation}
Thus
\begin{equation}
h_{33}=\frac{\alpha}{\nu}h_{11},
\qquad
h_{55}=\alpha\nu h_{11}.
\label{eq:diag-alpha}
\end{equation}

Equations \eqref{eq:first-offdiag}--\eqref{eq:diag-alpha}, together
with Hermiticity of $H$, show that all off-diagonal coefficients
vanish and that all diagonal coefficients are fixed by the single real
parameter $h_{11}$. Consequently every Hermitian matrix satisfying
both range constraints is of the form
\begin{align}
X
&=
h_{11}
\left(
f_1f_1^*
+\frac{\beta}{\mu}f_2f_2^*
+\frac{\alpha}{\nu}f_3f_3^*
+\beta^2 f_4f_4^*
+\alpha\nu f_5f_5^*
\right)
\nonumber\\
&=
h_{11}A_{\alpha,\beta,\mu,\nu}.
\end{align}

Thus the real vector space of Hermitian matrices satisfying both range
constraints is one-dimensional, and its positive semidefinite part is precisely
$\mathbb R_+\rho_{\alpha,\beta,\mu,\nu}$. Hence
\begin{equation}
\mathcal F(D,E)=\mathbb R_+\rho_{\alpha,\beta,\mu,\nu}.
\label{tde}
\end{equation}
Finally, let $P_{D^\perp}$ and $P_{E^\perp}$ be the orthogonal projections onto
$D^\perp$ and $E^\perp$. Define
\begin{equation}
F(X)=\Tr\!\left(\bigl[P_{D^\perp}+(P_{E^\perp})^\tau\bigr]X\right).
\end{equation}
For every $X\in T$,
\begin{equation}
F(X)=\Tr(P_{D^\perp}X)+\Tr(P_{E^\perp}X^\tau)\ge0.
\end{equation}
Moreover, $F(X)=0$ if and only if
\begin{equation}
\Ran X\subset D,
\qquad
\Ran X^\tau\subset E.
\end{equation}
Hence the zero set of $F$ in $T$ is precisely $\mathcal F(D,E)$.
Since we have shown (\ref{tde}), the ray generated by $\rho_{\alpha,\beta,\mu,\nu}$ is exposed.
\end{proof}

The above construction is consistent with the standard duality between
positive maps and bipartite states: indecomposable witnesses detect
PPT entangled states, while the associated range geometry organizes
the relevant faces of the PPT
cone~\cite{HaKyePark2003,HaKye2004,HaKyeExposed2011}. The endpoint
witness $W_*=W(\frac23,\frac23)$ also plays a distinguished role on
the witness side. Its dual face in the separable cone, its
nonoptimality, and an optimal refinement are analyzed in the next
section.
\section{The endpoint witness $W_*$: nonoptimality, optimal refinement, dual face, and exposedness}
\label{sec:wstar}
In this section we collect the witness-side properties completing the geometric picture.
This witness-side question should not be confused with the exposedness of the PPT-state rays proved above: the former concerns the cone of positive maps, whereas the latter concerns the PPT cone.
In this section we collect the witness-side properties completing the geometric picture. We prove that $W_*$ is not optimal, exhibit
a natural witness refinement, show that this refinement is optimal and
enlarges the detection region on the four-parameter PPT-edge family of
Sec.~\ref{sec:pptstates}. Then, we determine the dual face of $W_*$ in the
separable cone, and explain why the corresponding positive map does
not generate an exposed ray.  These questions belong to the broader facial-structure analysis of the separable and PPT cones; see, in particular,
Refs.~\cite{Lewenstein2000,AugusiakBaeCzekajLewenstein2011,HaKyeSPA2012,ChruscinskiSarbicki2014,Kye2011Faces,FacialStructures2012}.

\begin{proposition}
\label{prop:Wstar-not-optimal}
The endpoint witness $W_*:=W\!\left(\frac23,\frac23\right)$
 is not optimal.
\end{proposition}
\begin{proof}
Let $\Phi_*=\Phi_{2/3,2/3}$ be the corresponding positive map, and
define $\Psi(X)=V^\dagger X V$, where
\begin{equation}
V = \frac{1}{\sqrt{2}} \begin{pmatrix} 0 & 1 & 0 \\ 1 & 0 & 0 \\ 0 & 0
  & 0 \end{pmatrix}\,.
\end{equation}
Then $\Psi$ is a nonzero completely positive map. Consider $\widetilde\Phi:=\Phi_* - \frac23\,\Psi$.
We claim that $\widetilde\Phi$ is positive. It is enough to check positivity on
rank-one projectors $\ket{y}\!\bra{y}$, with $y=(a,b,c)^T$. Writing
\begin{equation}
u=|a|^2,\qquad v=|b|^2,\qquad q=|c|^2,\qquad s=u+v+q,
\label{sett}
\end{equation}
one finds
\begin{eqnarray}
\widetilde\Phi(\ket{y}\!\bra{y})
=
\frac 13 \begin{pmatrix}
u+q & \overline a\,b & 2a\,\overline c\\[1mm]
a\,\overline b & v+q & 0\\[1mm]
2\overline a\,c & 0 & s
\end{pmatrix}
\;.\label{reca}
\end{eqnarray}
Its diagonal entries are nonnegative. The nontrivial $2\times2$ principal minors
are
\begin{align}
\Delta_{12}&=\frac{qs}{9}\ge0,\\
\Delta_{13}&=\frac{(u-q)^2+v(u+q)}{9}\ge0,\\
\Delta_{23}&=\frac{s(v+q)}{9}\ge0,
\end{align}
and its determinant is
\begin{equation}
\frac{q\,(u-v-q)^2}{27}\ge0.
\end{equation}
Hence all principal minors are nonnegative, so
\begin{equation}
\widetilde\Phi(\ket{y}\!\bra{y})\succeq0
\qquad \forall\,y\in\mathbb C^3,
\end{equation}
and $\widetilde\Phi$ is positive. Therefore $\Phi_*$ admits subtraction of a
nonzero completely positive map while remaining positive, so $\Phi_*$ is not
optimal. Equivalently, $W_*$ is not optimal~\cite{Lewenstein2000,ChruscinskiSarbicki2014}.
\end{proof}

The proof singles out the refined witness
\begin{equation}
\widetilde W:=W_* - \frac23\,C_\Psi,
\label{eq:Wtilde-def}
\end{equation}
where $C_\Psi$ is the Choi matrix of $\Psi$.  We record its explicit form,
its action on $\rho_{\alpha,\beta,\mu,\nu}$, and its optimality.

\begin{proposition}[A refined optimal witness]
\label{prop:Wtilde-optimal}
Let $\widetilde W$ be defined by \cref{eq:Wtilde-def}. Then:
\begin{enumerate}
\item $\widetilde W$ is an optimal entanglement witness;
\item the Choi matrix of $\Psi$ is
\begin{equation}
C_\Psi=\frac12\,(e_2+e_4)(e_2+e_4)^*;
\label{eq:Cpsi-explicit}
\end{equation}
\item for every $\alpha,\beta,\mu,\nu>0$,
\begin{equation}
\Tr(\widetilde W\,\rho_{\alpha,\beta,\mu,\nu})
=
\frac{(\alpha-1)^2+(\beta-1)^2-1
+\alpha(\nu+\nu^{-1}-2)}{3N_{\alpha,\beta,\mu,\nu}};
\label{eq:Wtilde-action}
\end{equation}
\item if $\mathcal D(W)$ denotes the set of parameters
$(\alpha,\beta,\mu,\nu)$ for which $W$ detects the state
$\rho_{\alpha,\beta,\mu,\nu}$, then
\begin{equation}
\mathcal D(W_*)\subsetneq\mathcal D(\widetilde W).
\label{eq:strict-inclusion-detect}
\end{equation}
\end{enumerate}
\end{proposition}

\begin{proof}
Since $\widetilde\Phi:=\Phi_*-\frac23\,\Psi$ is positive by the proof of
\cref{prop:Wstar-not-optimal}, its Choi matrix $\widetilde W$ is block-positive.
We first compute its action on the family $\rho_{\alpha,\beta,\mu,\nu}$.

A direct computation from $(I \otimes V^\dag ) |\psi \rangle \langle
\psi | (I\otimes V)$ gives \cref{eq:Cpsi-explicit}. Therefore,
\begin{equation}
\Tr\bigl(C_\Psi\,\rho_{\alpha,\beta,\mu,\nu}\bigr)
=\frac12\bigl[(\rho_{\alpha,\beta,\mu,\nu})_{22}
+(\rho_{\alpha,\beta,\mu,\nu})_{44}
+2 (\rho_{\alpha,\beta,\mu,\nu})_{24}\bigr]
=\frac{\beta}{2N_{\alpha,\beta,\mu,\nu}}\bigl(\mu+\mu^{-1}-2\bigr),
\end{equation}
where we used the explicit matrix form of $\rho_{\alpha,\beta,\mu,\nu}$ from
Sec.~\ref{sec:pptstates}. Substituting this into
$\widetilde W=W_*-\frac23 C_\Psi$ and using the formula for
$\Tr(W_*\rho_{\alpha,\beta,\mu,\nu})$ from the theorem in
Sec.~\ref{sec:pptstates} yields the exact cancellation of the
$\beta(\mu+\mu^{-1}-2)$ term and gives \cref{eq:Wtilde-action}.
In particular,
\begin{equation}
\Tr\bigl[\widetilde W\,\rho_{1,1,3,1}\bigr]
=\frac{-1}{3N_{1,1,3,1}}<0.
\end{equation}
Thus $\widetilde W$ is not positive semidefinite. Since it is block-positive,
it is an entanglement witness.

The detection condition for $W_*$ is given by \cref{condd}, whereas by
\cref{eq:Wtilde-action} the detection condition for $\widetilde W$ is
\begin{eqnarray}
(\alpha-1)^2+(\beta-1)^2+\alpha(\nu+\nu^{-1}-2)<1\;.\label{cond2}
\end{eqnarray}
Since $\mu+\mu^{-1}-2\ge0$,  \cref{condd} implies \cref{cond2}, so
$\mathcal D(W_*)\subseteq\mathcal D(\widetilde W)$. The inclusion is strict;
for instance, $(\alpha,\beta,\mu,\nu)=(1,1,3,1)$ satisfies the second
inequality but not the first. This proves \cref{eq:strict-inclusion-detect}. The improved detection formula \cref{eq:Wtilde-action}
also makes the $\mu$-independence transparent.

To prove optimality, we use the spanning criterion for entanglement
witnesses; see, e.g.,
Refs.~\cite{Lewenstein2000,ChruscinskiSarbicki2014}. Let
$y=(a,b,c)^T$, set Eq. (\ref{sett}), and recall $\widetilde\Phi(\ket
y\!\bra y)$ of \cref{reca}. A direct check shows that the following
two families lie in its kernel:
\begin{enumerate} \item if $c=0$, then
\begin{equation} x=(b,-a,0)^T \qquad\Longrightarrow\qquad
\widetilde\Phi(\ket y\!\bra y)\,x=0; \end{equation}
\item if
  $|a|^2=|b|^2+|c|^2$, then $a\neq0$ for nonzero $y$ and, after fixing
the global phase of $y$ so that $a>0$, 
\begin{equation} x=(a,-\overline b,-c)^T
\qquad\Longrightarrow\qquad \widetilde\Phi(\ket y\!\bra y)\,x=0.
\end{equation} 
\end{enumerate}
Hence, from the identity \cref{eq:ChoiM}, the corresponding product
vectors $\overline y\otimes x$ are annihilated by $\widetilde W$.

We now exhibit nine such product vectors whose span is all of
$\mathbb C^3\otimes\mathbb C^3$:
\begin{align}
z_1&=\overline{(0,1,0)}\otimes(1,0,0)=e_4,\\
z_2&=\overline{(1,0,0)}\otimes(0,-1,0)=-e_2,\\
z_3&=\overline{(1,1,0)}\otimes(1,-1,0)=e_1-e_2+e_4-e_5,\\
z_4&=\overline{(1,i,0)}\otimes(i,-1,0)=i\,e_1-e_2+e_4+i\,e_5,\\
z_5&=\overline{(1,0,1)}\otimes(1,0,-1)=e_1-e_3+e_7-e_9,\\
z_6&=\overline{(1,0,-1)}\otimes(1,0,1)=e_1+e_3-e_7-e_9,\\
z_7&=\overline{(1,0,i)}\otimes(1,0,-i)=e_1-i\,e_3-i\,e_7-e_9,\\
z_8&=\overline{(\sqrt2,1,1)}\otimes(\sqrt2,-1,-1)
   =2e_1-\sqrt2\,e_2-\sqrt2\,e_3+\sqrt2\,e_4-e_5-e_6+\sqrt2\,e_7-e_8-e_9,\\
z_9&=\overline{(\sqrt2,1,i)}\otimes(\sqrt2,-1,-i)
   =2e_1-\sqrt2\,e_2-i\sqrt2\,e_3+\sqrt2\,e_4-e_5-i\,e_6-i\sqrt2\,e_7+i\,e_8-e_9.
\end{align}
From $z_1$ and $z_2$ we obtain $e_4$ and $e_2$. Then
\begin{equation}
z_3+e_2-e_4=e_1-e_5,
\qquad
z_4+e_2-e_4=i(e_1+e_5),
\end{equation}
so $e_1,e_5$ lie in the span. Next,
\begin{equation}
z_5-e_1=-e_3+e_7-e_9,\qquad
z_6-e_1=e_3-e_7-e_9,\qquad
z_7-e_1=-i\,e_3-i\,e_7-e_9.
\end{equation}
Adding the first two gives $-2e_9$, hence $e_9$ belongs to the span; then
$e_3-e_7$ and $e_3+e_7$ are obtained from the remaining combinations, so
$e_3,e_7$ also belong to the span. Finally,
\begin{equation}
z_8-\bigl(2e_1-\sqrt2 e_2-\sqrt2 e_3+\sqrt2 e_4-e_5+\sqrt2 e_7-e_9\bigr)
=-(e_6+e_8),
\end{equation}
and
\begin{equation}
z_9-\bigl(2e_1-\sqrt2 e_2-i\sqrt2 e_3+\sqrt2 e_4-e_5-i\sqrt2 e_7-e_9\bigr)
=i(-e_6+e_8).
\end{equation}
Hence $e_6,e_8$ also belong to the span. Therefore
\begin{equation}
\spanop\{z_1,\dots,z_9\}=\mathbb C^3\otimes\mathbb C^3.
\end{equation}
Thus the set of product vectors annihilated by $\widetilde W$ spans the whole
space, so $\widetilde W$ has the spanning property. By the standard spanning
criterion, $\widetilde W$ is optimal.
\end{proof}

\begin{corollary}[Independence from the $\mu$-deformation]
\label{cor:mu-independence}
For every $\alpha,\beta,\nu>0$ and $\mu>0$ one has
\begin{equation}
\Tr(\widetilde W\,\rho_{\alpha,\beta,\mu,\nu})
=
\Tr(\widetilde W\,\rho_{\alpha,\beta,1,\nu}).
\end{equation}
In particular, the detection condition for $\widetilde W$ is independent of
$\mu$ and is given by \cref{cond2}.
\end{corollary}

\begin{proposition}[Exact dual face of the endpoint witness]
\label{prop:dual-face-Wstar-appendix}
Let $S$ denote the cone of separable states on $\mathbb C^3\otimes\mathbb C^3$,
and define
\begin{equation}
W_*':=\{\sigma\in S:\ \Tr(W_*\sigma)=0\}.
\end{equation}
Then
\begin{equation}
W_*'=
\operatorname{cone}
\Bigl\{\,
\ket{z(\beta,\gamma)}\!\bra{z(\beta,\gamma)}:
|\beta|^2+|\gamma|^2=1
\Bigr\},
\end{equation}
where $\operatorname{cone}$ denotes the conic hull, i.e., finite nonnegative
linear combinations of the displayed projectors, and
\begin{equation}
z(\beta,\gamma)=\overline{y(\beta,\gamma)}\otimes x(\beta,\gamma),
\end{equation}
with
\begin{equation}
y(\beta,\gamma)=\frac1{\sqrt2}(1,\beta,\gamma)^T,
\qquad
x(\beta,\gamma)=(1,-\overline{\beta},-\gamma)^T.
\end{equation}
\end{proposition}
\begin{proof}
Let $\Phi_*=\Phi_{2/3,2/3}$ be the positive map whose Choi matrix is $W_*$. For
any product vector $\overline y\otimes x$, the Choi pairing gives
\begin{equation}
\langle \overline y\otimes x | W_* | \overline y\otimes x\rangle
=\langle x |\Phi_*(\ket{y}\!\bra{y}) | x\rangle.
\end{equation}
Hence the zero product vectors of $W_*$ are those for which
$x\in\ker\Phi_*(\ket{y}\!\bra{y})$. Writing $y=(a,b,c)^T$ with
$|a|^2+|b|^2+|c|^2=1$, one has
\begin{equation}
\Phi_*(\ket{y}\!\bra{y})
=\frac13
\begin{pmatrix}
1 & 2\overline a\,b & 2a\,\overline c\\
2a\,\overline b & 1 & 0\\
2\overline a\,c & 0 & 1
\end{pmatrix}.
\end{equation}
The determinant computation gives
\begin{equation}
\det \Phi_*(\ket{y}\!\bra{y})
=\frac1{27}\bigl(1-4|a|^2|b|^2-4|a|^2|c|^2\bigr)
=\frac1{27}\bigl(1-4|a|^2(1-|a|^2)\bigr),
\end{equation}
which vanishes if and only if $|a|^2=\frac12$. Up to a global
phase, every such $y$ is of the form
\begin{equation}
y(\beta,\gamma)=\frac1{\sqrt2}(1,\beta,\gamma)^T,
\qquad |\beta|^2+|\gamma|^2=1.
\end{equation}
For this choice, one checks directly that
\begin{equation}
\Phi_*\bigl(\ket{y(\beta,\gamma)}\!\bra{y(\beta,\gamma)}\bigr)\,
x(\beta,\gamma)=0,
\qquad x(\beta,\gamma)=(1,-\overline\beta,-\gamma)^T.
\end{equation}
Therefore the zero product vectors of $W_*$ are precisely
\begin{equation}
z(\beta,\gamma)=\overline{y(\beta,\gamma)}\otimes x(\beta,\gamma),
\qquad |\beta|^2+|\gamma|^2=1.
\end{equation}
Since the dual face $W_*'$ consists exactly of the separable states annihilated
by $W_*$, it is the conic hull of the corresponding rank-one product
projectors.
\end{proof}

\begin{remark}[What the dual-face description does and does not imply]
\label{rem:Wstar-not-exposed}
The explicit description of $W_*'$ should be viewed in the context of
the general theory of faces of the separable cone induced by
PPT-related geometry; compare
Refs.~\cite{Kye2011Faces,FacialStructures2012}. Proposition~\ref{prop:dual-face-Wstar-appendix}
determines the dual face $W_*'=\{\sigma\in
S:\Tr(W_*\sigma)=0\}$ exactly inside the separable cone $S$. This
identifies all separable states annihilated by the witness
$W_*$. However, map-side exposedness is a statement about the face
generated by $W_*$ in the cone of positive maps, and would require the
bidual equality
\begin{equation}
(W_*')'=\mathbb R_+W_*.
\end{equation}
For general sufficient criteria ensuring that a positive map generates an exposed
ray, see Ref.~\cite{ChruscinskiSarbicki2012}. Thus the exact description of
$W_*'$ does not by itself prove exposedness.

In the present case one can say more: $W_*$ is in fact \emph{not} exposed.
Indeed, exposed rays are extremal, and any extremal positive map that is not
completely positive is necessarily optimal~\cite{Lewenstein2000,ChruscinskiSarbicki2014}.
But \cref{prop:Wstar-not-optimal} shows that $W_*$ is not optimal, while $W_*$
is not completely positive since $(2/3,2/3)$ lies outside the completely
positive square $[0,1/3]\times[0,1/3]$. Hence $W_*$ is neither extremal nor
exposed.
\end{remark}

\section{Conclusion}
We have introduced an exactly solvable family of sparse unital and
trace-preserving positive maps $\Phi(w,z)$ on qutrits.  The two parameters
control two distinct coherence channels of the associated Choi witness, and the
absence of a third coherence sector makes the construction analytically
tractable.  In the uniform bistochastic slice the full two-parameter phase
diagram is exact: positivity occupies $[0,2/3]\times[0,2/3]$, complete positivity
occupies $[0,1/3]\times[0,1/3]$, and indecomposability begins precisely beyond
the quarter-circle boundary \eqref{uno}.

The same sparse structure also gives explicit PPT entangled states detected by
the corresponding witnesses.  In particular, the endpoint witness
$W_*=W(2/3,2/3)$ detects a four-parameter family of PPT edge states of rank type
$(5,5)$, with an analytic detection condition.  The range and kernel constraints
for these states determine a one-dimensional face of the PPT cone, proving that
the constructed PPT-state rays are exposed.  Thus the family provides a concrete
qutrit setting in which PPT entanglement, edge-state geometry, and exposed faces
can be displayed explicitly.

On the witness side, the endpoint $W_*$ separates indecomposability from
optimality: it detects PPT entanglement but is not optimal.  We identified a
simple completely positive subtraction that preserves positivity and yields an
optimal refinement $\widetilde W$.  This refinement has a strictly larger
analytic detection region on the constructed PPT-edge family, showing explicitly
how optimization improves the witness while preserving tractability.

The higher-dimensional extension is conceptually natural but analytically more
demanding. In dimension $d$, additional active coherence sectors introduce a
rapidly increasing number of independent weights, larger coupled Choi blocks,
and simplex minimizations with many variables. The endpoint reduction that
makes \cref{eq:S-min} affine in one simplex coordinate is then generally lost,
and explicit $P+Q^\tau$ decompositions must satisfy a substantially larger set
of coupled positivity constraints. A systematic route is to impose diagonal
unitary or orthogonal covariance, which organizes the Choi matrix into invariant
sectors and converts positivity, separability, and entanglement breaking into
cone-membership problems \cite{SinghNechita2021}. Recent pairwise
copositive/completely-positive and pairwise-decomposable cone constructions give
related characterizations of positivity and decomposability for broad covariant
families and yield graph-parameterized examples in arbitrary dimension
\cite{GulatiNechitaPark2025}. Adapting those symmetry and cone-theoretic tools to
sparse channel-selective patterns is a natural direction for future work.

The family can be used as a benchmark for numerical and variational searches for
positive maps and PPT entangled states, since the exact positivity,
decomposability, and detection thresholds are known.  Natural extensions include
perturbations away from the uniform bistochastic slice, higher-dimensional
sparse constructions, and the search for sparse families in which more coherence
sectors can be tuned while retaining an analytic PPT-state geometry.

\section*{Data Availability}
No external data were used in this work.  The figures and parameter regions are
generated directly from the analytic expressions given in the text.

\end{document}